\documentclass[aps,pra,twocolumn,groupedaddress,amsmath,amssymb,10pt,nofootinbib]{revtex4-1}

\usepackage[dvipdfm]{hyperref} 

\usepackage{amssymb} \usepackage{amsmath} \usepackage{amsthm}
\usepackage{graphicx} \usepackage{amstext}
\usepackage[utf8]{inputenc}
\usepackage{mathtools}

\usepackage{amsmath}
\usepackage{amssymb}
\usepackage{amsthm}

\usepackage{fontenc}

\usepackage{amsmath,amssymb,amsthm, graphicx, array,epsfig, exscale, dsfont, amsfonts, verbatim, fancyhdr, natbib, bbm,mathrsfs,textcomp,listings, epsfig}

\usepackage[usenames,dvipsnames]{color}
\usepackage[capitalize]{cleveref}

\theoremstyle{plain}
\newtheorem{theorem}{Theorem}[section]
\newtheorem{lemma}[theorem]{Lemma}

\newtheorem{proposition}[theorem]{Proposition}
\theoremstyle{definition}
\newtheorem{definition}[theorem]{Definition}
\newtheorem{remark}[theorem]{Remark}

\newcommand{\one}{\mbox{$1 \hspace{-1.0mm}  {\bf l}$}}

\newcommand{\bra}[1]{\left\langle{#1}\right\vert}
\newcommand{\ket}[1]{\left\vert{#1}\right\rangle}

\newcommand{\myeps}{\varepsilon}
\newcommand{\epsbar}{\overline{\myeps}}

\newcommand{\erk}{\hat{\myeps}}
\newcommand{\etr}{\overline{\myeps}-\hat{\myeps}}
\newcommand{\soproj}{\Pi}

\DeclareMathOperator{\tr}{tr}
\newcommand{\ketbra}[2]{{|#1\rangle\!\langle#2|}}
\newcommand{\pr}[1]{{\ketbra{#1}{#1}}}
\newcommand{\ball}[2]{{\cal B}^{#1}\left(#2\right)}
\newcommand{\dist}[1]{\frac{1}{2}\left|\left|#1\right|\right|_1} 
\newcommand{\eps}{\myeps}

\newcommand{\ren}[3]{S_{#1}^{#2}\left(#3\right)}
\newcommand{\leak}{\mathrm{leak}_{\mathrm{EC}}} 
\newcommand{\cqball}[2]{{\cal B}^{#1}_{\textrm{cq}}\left(#2\right)}

\newcommand{\rhoAB}{\rho_{AB}}

\date{\today}

\begin{document}

\title{QKD with finite resources: secret key rates via R\'enyi entropies}
\author{Silvestre Abruzzo}
\email{abruzzo@thphy.uni-duesseldorf.de}
\author{Hermann Kampermann}
\author{Markus Mertz}
\author{Dagmar Bru{\ss}}
\affiliation{Institute for Theoretical Physics III, Heinrich-Heine-universit\"at D\"usseldorf, 40225 D\"usseldorf, Germany.}

\begin{abstract}
A realistic Quantum Key Distribution (QKD) protocol necessarily deals with  finite resources, such as the number of signals exchanged by the two parties. We derive a bound on the secret key rate which is expressed as an optimization problem over R\'enyi entropies. Under the assumption of collective attacks by an eavesdropper, a computable estimate of our bound for the six-state protocol is provided. This bound leads to improved key rates in comparison to previous results.
\end{abstract}
\maketitle

%%%%%%%%%%%%%%%%%%%%%%
\section{Introduction}
Quantum Key Distribution (QKD) is a method for transmitting a secret key between two partners.   Since its initial proposal \cite{bennett1984quantum} QKD has reached maturity through many theoretical developments and experimental realizations. Moreover, in the last few years  QKD has entered  the commercial market \cite{shields2007key} and small QKD networks were realized \cite{tokyo:qkd,sudafrica:qkd}.

A significant figure of merit in QKD is the \emph{secret key rate}, i.e. the ratio between  the length of the secure key and  the initial number of signals. There is a big difference between the key rate calculated under the assumption  that the key is composed of an infinite number of bits, and a key in real applications, where the number of bits is finite. In recent years a new paradigm for security in the finite key setting was developed \cite{Renner:2005pi,Renner:kl, renner05, Christandl:ye}.  However, the complexity of the entropic quantities involved in the formalism only permits to find bounds on the optimal quantities, which leads to much lower key rates for a small number of signals with respect to the asymptotic ones. 

To our knowledge the first work dealing with finite key corrections is \cite{mayers2001unconditional}. The currently used framework for finite-key analysis was developed in \cite{renner05, Renner:kl}.  The bound proved in \cite{renner05} was used by T.~Meyer~et al. \cite{tim06} to calculate the key rate in the finite-key scenario. In \cite{Scarani:2008ve, Scarani:2008ys} security bounds for the BB84 and the six-state protocol were provided using an easily calculable  bound for the smooth min-entropy.  Recently, many efforts were done for improving the bounds on the secret key rates for a finite amount of resources, e.g. using the connection between the min-entropy and the guessing probability \cite{Koenig:2008qf, bratzik2010min}. So far the secret key rates provided are only proven to be secure for collective attacks.  A possible approach for providing security against coherent attacks using the results against collective attacks can be obtained by post-selections techniques \cite{Christandl:2009dq, RennerPostNato, sheridan2010finite} or the exponential de Finetti theorem \cite{Renner:kl}. A recent technique is given by  uncertainty relations for the smooth min-entropy \cite{Tomamichel2010, TomLim2011}. This last approach is very promising because it provides an easily calculable tight bound on the key rate even for coherent attacks, however it is not easily applicable to the six-state protocol.
A step in the direction of considering more practical issues  in addition to finite-key corrections (BB84 with and without decoy states and entanglement-based  implementations)  was provided in \cite{Scarani:2009wr, li2009security}. 
\\
\\
In this paper, we present a bound on the achievable key length for the  six-state protocol. The presented bound is resorted from \cite[Lemma 9]{renner05}, where it is used for bounding the key length in terms of smooth R\'enyi entropies. 
We calculate explicitly the presented bound under the assumptions of collective attacks and the depolarizing channel. The calculated key rates  lead for small number of signals  to better key rates  than the  bounds derived in \cite{tim06,Scarani:2008ve, Scarani:2008ys}.
\\
\\
The paper  is organized as follows. In \cref{Sec:DescProt} we present the protocol we are going to study. In \cref{sec:def} we introduce definitions and our notation. In \cref{Sec:security}  we explain the approach developed in this paper and we show how to estimate the proposed bound for the achievable key rate. In \cref{sec:result}  we compare the proposed bound with other relevant bounds present in the literature. \cref{sec:conclusion} contains the conclusions. In the appendices we prove additional results used in the paper.

\section{\label{Sec:DescProt} Description of the protocol}

In this paper we consider the entanglement-based version of the  \emph{six-state protocol} \cite{bruss1998optimal,6stategisin}. The protocol consists of the following steps.

{\bf State preparation and distribution:} Alice prepares N entangled Bell states and distributes one part of each pair to Bob. We assume that Eve performs at most a collective attack, i.e. the adversary acts on each of the signals independently and identically.

{\bf Reduction to Bell-diagonal form:}
Alice and Bob apply randomly and simultaneously one of the operators $\{\one, \sigma_X, \sigma_Y, \sigma_Z\}$ and as a result they obtain a Bell-diagonal state  the entries of which are directly connected with the Quantum Bit Error Rates (QBER), \cite{scarani:2009} and  \cref{sec:usedoper}. 

{\bf Sifting and Measurement:}
Alice and Bob measure at  random one of the three Pauli operators. The Pauli operators are chosen with different probabilities. We consider that $\sigma_{X}$ and $\sigma_{Y}$ are chosen with the same probability and that $\sigma_{Z}$ is chosen such that $\mathrm{Pr}(\sigma_{Z})\geq \mathrm{Pr}(\sigma_{X})$. This biased setting \cite{lo2005efficient} is advantageous in terms of sifting.  At the end of the measurement process, Alice and Bob broadcast  the choice of the bases through  the classical channel and  discard the results coming from a different choice of the measurement basis. We call $n'=n'_{X}+n'_{Y}+n'_{Z}$ the length of the sifted key shared by Alice and Bob, where $n'_{i}$ with $i=X, Y, Z$ is the remaining number of signals when both Alice and Bob measure $\sigma_{X}, \sigma_{Y},\sigma_{Z}$.

{\bf Parameter estimation:}
Parameter estimation (PE) permits to measure the amount of errors on the key, which in the security analysis are assumed to be introduced via Eve's eavesdropping. In the six-state protocol three bases are used for the measurement and therefore a QBER in each direction is calculated by Alice and Bob. Practically speaking, Alice broadcasts for each basis $m_{i}<n'_{i}$ bits of the sifted key on the classical channel. Bob compares these outcomes with his corresponding outcomes and calculates the QBERs $e_{m}^{i}$ as the ratio between the number of discordant positions and the length of the transmitted strings. In general $e_{m}^{X}\neq e_{m}^{Y}\neq e_{m}^{X}$. For calculating explicitly the bound that we are going to propose we use the biggest QBER as measured QBER, denoted as $e_{m}$. Note that it is possible to introduce additional symmetrizations \cite{lo2001proof,kraus:2005kx}  that reduce the initial state to a state described by only one parameter: the QBER.  However those symmetrizations require additional experimental means that could be difficult to implement.

 The remaining $n:=n'-m_{X}-m_{Y}-m_{Z}$ bits will be used for the extraction of the key. The QBER $e$ is bounded by parameter estimation developed in \cite{Scarani:2009wr,Scarani:2008ve,Scarani:2008ys,bratzik2010min}.
The parameter $\myeps_{PE}$ represents the probability that  we underestimated the real QBER.

The QBER of the key $e$ with probability $1-\myeps_{PE}$ is such that\cite{Scarani:2009wr,Scarani:2008ve,Scarani:2008ys,bratzik2010min}
\begin{equation}
\label{eq:pa:boundQBER}
 e \leq  e_{m} +  2\zeta\left(\myeps_{PE},m\right)
\end{equation}
with
\begin{equation}
\zeta(\myeps_{PE},m):=\sqrt{\frac{\ln{\left(\frac{1}{\myeps_{PE}}\right)}+2\ln{(m+1)}}{8m}}.
\end{equation}

{\bf Error correction:}
 Alice and Bob hold correlated classical bit strings $X^n$ and $Y^n$. The purpose of an error correction (EC) protocol is to create a fully correlated string, while leaking only a small amount of information to an adversary. In the following, we will consider realistic error correction protocols. The number of bits leaked during the classical communication to an eavesdropper is given by \cite{Scarani:2009wr, Scarani:2008ve}
 \begin{equation}
\label{eq:leakec}
 \leak = f_{\mathrm{EC}}nh(e)+\log{\left(\frac{2}{\eps_{EC}}\right)},
\end{equation} where $f_{\mathrm{EC}}\gtrsim1$ depends on the used EC protocol, $h(e)$ is the binary Shannon entropy, i.e.  $h(e)=-e\log{e}-(1-e)\log{(1-e)}$ and $e$ is the QBER. Here, $\myeps_{EC}$ is the probability that Alice's and Bob's strings differ after the error correction step.

 {\bf Privacy amplification:} 
  Let Alice and Bob hold a perfectly correlated bit string $X^{n}$, on which Eve might have some information.  The purpose of privacy amplification is to shrink the length of  $X^{n}$ in order to reduce  Eve's information on the resulting string. 

Practically, Alice chooses at random a two-universal hash function (Definition~\ref{defi:hash} in \cref{proof:main:theo})  and communicates it to Bob. 

\vskip 1cm

\section{Definitions and Notation\label{sec:def}}
The set of quantum states, which are normalized positive semidefinite bounded operators, will be represented by $S\left(\mathcal{H}\right)$, where $\mathcal{H}$ stands for a finite-dimensional Hilbert space. In the following $\rho_A(\rho_B)$ belongs to the set of bounded operators which act on the Hilbert space $\mathcal{H}^A(\mathcal{H}^B)$. For a given state $\rhoAB$, the states $\rho_A$, $\rho_B$ are defined via the partial trace, i.e. $\rho_A:=\tr_{B}\rho_{AB}$ and $\rho_B:=\tr_{A}\rho_{AB}$.

In this paper, we will consider R\'enyi entropies, which are a generalization of the Von Neumann entropy.

\begin{definition}(R\'{e}nyi entropies\cite{renyi1961measures, renner05}) Let $\alpha \in \mathbb{R} \cup \{\infty\}$ and $\rho, \sigma \in \mathrm{S}\left(\mathcal{H}\right)$. The R\'{e}nyi entropy of order $\alpha$ is defined as \footnote{$\log:=\log_2$} 
   \begin{equation}
   \ren{\alpha}{}{\rho}:=\frac{1}{1-\alpha}\log{\left(\tr{\left(\rho^{\alpha}\right)}\right)}.
   \end{equation} In particular, we get
   \begin{eqnarray}
   \ren{0}{}{\rho}           &=& \log{\left(\mathrm{rank}(\rho)\right)} \\
   \ren{2}{}{\rho}           &=& -\log{\left(\tr({\rho^2})\right)} \\
   \ren{\infty}{}{\rho}      &=& -\log{\left(\lambda_{\mathrm{max}}\left(\rho\right)\right)}
   \end{eqnarray} where $\lambda_{\mathrm{max}}\left(\rho\right)$ is the maximal eigenvalue of $\rho$.
\end{definition}

Another useful quantity  is the smooth R\'enyi entropy, which is the R\'enyi entropy optimized on a set of operators which are $\myeps$-close to the operator involved in the actual computation. 
 We define an $\eps{}$-environment via the trace-distance in the following way \cite{renner05}:
  \begin{definition}($\myeps$-environment) Let $\eps{}\geq 0$ and $\rho \in S\left(\mathcal{H}\right)$, then
   \begin{equation}
   \ball{\eps{}}{\rho}:=\left\lbrace \sigma \in S\left(\mathcal{H}\right): \dist{\sigma -\rho}\leq \eps{} \right\rbrace
   \end{equation}
  where $\left|\left|A\right|\right|_1=\tr \sqrt{AA^\dagger}$.
  \end{definition}
  
  \begin{definition}\label{def:smoothrenent}The smooth R\'{e}nyi entropy of order $\alpha$ is defined (following \cite{renner05}) as   \begin{equation}
   \ren{\alpha}{\eps{}}{\rho}:=\frac{1}{1-\alpha}\inf_{\sigma \in \ball{\frac{\eps{}}{2}}{\rho}}\log{\left(\tr{\left(\sigma^{\alpha}\right)}\right)}.
   \end{equation}
 \end{definition}
 
 The main result presented in this paper will be expressed as an optimization problem on a \emph{classical-quantum $\myeps$-environment} of a certain operator.  
 
 \begin{definition}\label{def:cqstates} (Classical-quantum(cq)-state) Let $\{\ket{x}\}$ be an orthonormal basis of $\mathcal{H}^{X}$ and moreover let $\mathcal{H}^{A}$ be a generic Hilbert space. We define the state $\rho_{XA}$ which is classical on $\mathcal{H}^{X}$ and quantum on $\mathcal{H}^{A}$ as the state
\begin{equation*}
\rho_{XA}=\sum_{x}P_{X}(x)\ket{x}\bra{x}\otimes\rho_{A}^{x},
\end{equation*}
where $\rho_{A}^{x}\in\mathcal{S}(\mathcal{H}^{A})$ and $P_{X}(x)$ is a classical probability distribution.
\end{definition}
Finally, we define the  \emph{classical-quantum $\myeps$-environment} as the space
\begin{align*}
\cqball{\myeps}{\rho_{XA}}&:=\{\sigma_{XA}\in\ball{\myeps}{\rho_{XA}}:\\ &\quad\quad\quad
\sigma_{XA}=\sum_{x}P_{X}(x)\ket{x}\bra{x}\otimes\sigma_{A}^{x}\},
\end{align*}
where $\sigma_{A}^{x}\in\mathcal{S}(\mathcal{H}^{A})$ and $P_{X}(x)$ is a classical probability distribution.
Finally, we recall the composable definition of security introduced by Renner in \cite{renner05}. For additional details see \cite{muller2009composability}.

\begin{definition}\label{def:securitydef} Let $\rho_{KE}$ be the cq-state describing the classical key $K$ of length $\ell$, distilled at the end of a QKD protocol, correlated with the quantum states of the eavesdropper $\rho_{E}$. The state $\rho_{KE}$  is said to be $\myeps$-secure if
\begin{equation}
\label{def:security}
\frac{1}{2}\|\rho_{KE}-\frac{1}{2^\ell}\one\otimes\rho_{E'}\|_{1}\leq\myeps,
\end{equation}
\end{definition}
where $\rho_{E'}$ is the quantum state of an eavesdropper not correlated with the key.

In the literature several bounds on an $\myeps$-secure key length \cite{renner05, Renner:kl, Tomamichel:2010fk} were presented.

\section{ \label{Sec:security} Bound on the achievable key length}

The following bound was inspired by \cite[Theorem 4]{renner05} where it was used as a bridge for providing an analogous bound in terms of smooth R\'enyi entropies.

 \begin{theorem}\label{theorem:PADEP}Let $\rho_{X^{n}E^{n}}$ be the cq-state describing Alice's bitstring $X^{n}$ as well as Eve's quantum information represented by $\rho_{E^n}$. Let $\epsbar, \myeps_{PA}\geq0$.  If the length $\ell$ of the key is such that
   \begin{align}
\label{for:depopt}
   \ell \leq &\sup_{\sigma_{X^nE^n}\in \cqball{\frac{\overline{\eps}}{2}}{\rho_{X^nE^n}}}{\left(\ren{2}{}{\sigma_{X^nE^n}}-\ren{0}{}{\sigma_{E^n}}\right)} \nonumber\\&\quad\quad\quad -\leak +2\log\left(2\myeps_{PA}\right),
   \end{align}
  then the key is $\epsbar+\myeps_{PA}$-secure.
    \end{theorem}

  \emph{Sketch of Proof:} In the following we give an idea of the proof which follows the lines of \cite{renner05, Renner:kl}. For all details see \cref{proof:main:theo}.   We first prove that $\ell$ can be chosen such that
  \begin{align}\label{proof:step1}
   \ell \leq &\sup_{\sigma_{X^nE^nC}\in \ball{\frac{\epsbar}{2}}{\rho_{X^nE^nC}}}{\left(\ren{2}{}{\sigma_{X^nE^nC}}-\ren{0}{}{\sigma_{E^nC}}\right)}+ \nonumber \\&\quad\quad\quad+2\log\left(2\myeps_{PA}\right),
  \end{align}
  where the additional random variable $C$ is associated with the probability distribution of transcripts of the EC protocol. Then we will ''extract'' the leakage term using the data processing inequality and the subadditivity of the R\'enyi entropies.
  \qed
  
The bound in \cref{for:depopt} is related to the bound calculated in \cite{tim06} because it involves optimizations on R\'enyi entropies. However, in \cite{tim06} the two R\'enyi entropies are optimized independently and here we have a combined optimization problem. This additional constraint is mitigated by the fact that in our bound we optimize over a bigger environment than the one used in  \cite{tim06}, more precisely $\myeps'=\frac{\overline{\myeps}^2}{2}$ where $\myeps'$ is the environment used for the smooth R\'enyi entropies in \cite{tim06}.

\subsection{Lower bound of \cref{theorem:PADEP} using smooth R\'enyi entropies}
In this section, we present a lower bound for the key length presented in \cref{theorem:PADEP}. The optimization problem involved in equation \cref{for:depopt}  is exponentially complex  because the dimension of the involved operators scales with $n$, that is the length of the string used for extracting the key. For reducing the complexity of the problem we consider the symmetric six-state protocol. For this protocol the number of different eigenvalues in $\rho_{XE}^{\otimes n}$ scales polynomially with $n$ \cite{tim06}, therefore as done in \cite{tim06}, it is possible to concentrate on optimizing the eigenvalue distribution of $\sigma_{X^n E^n}$. However it is not clear how to find the eigenvalue distribution of $\sigma_{E^n}$ from the one of $\sigma_{X^n E^n}$ in such a way that it is possible to perform computations for big $n$. In the following we present a lower bound on \cref{theorem:PADEP} expressed in terms of the smooth R\'enyi entropy of order zero and a \emph{modified smooth R\'enyi entropy of order two} that we will denote as $\overline{S}_{2}^{\myeps}(\rho_{XE}^{\otimes n})$. This last entropy permits to bound the eigenvalues of $\sigma_{E^n}$ for a given  $\sigma_{X^n E^n}$.  From the numerical point of view the deviation from $S_{2}^{\myeps}(\rho_{XE}^{\otimes n})$  is negligible.

\subsubsection{Modified Smooth R\'enyi entropy of order two}
Let  $\rho_{XE}^{\otimes n}$ be the operator describing  Alice's classical string of $n$ bits correlated with the operator $\rho_{E}^{\otimes n}$ held by Eve. The operator $\rho_{XE}^{\otimes n}$  is constructed by a direct sum of $2^n$ blocks which have the same eigenvalues (see \cref{sec:usedoper} for additional details).

\begin{definition}\label{def:modsmoothrentwo}
The modified smooth R\'enyi entropy of order two of the operator $\rho_{XE}^{\otimes n}$ is defined by
\begin{equation}
\label{eq:def:modsmoothrentwo}
\overline{S}_{2}^{\myeps}(\rho_{XE}^{\otimes n}):=S_{2}(\tau_{X^n E^n}),
\end{equation}
where the operator $\tau_{X^n E^n}$ has the following properties:
\begin{enumerate}
\item  $\tau_{X^n E^n}$ has the following form
\begin{equation}
\tau_{X^n E^n}:=\frac{1}{2^{n}}\sum_{x=0}^{2^{n}-1}\ket{x}\bra{x}\otimes\tau_{E^n}^{x},
\end{equation}
where $\{\ket{x}\}$ is the basis in which $\rho_{XE}^{\otimes n}$ is a classical-quantum state (Definition~\ref{def:cqstates}). Moreover each of the $\{\tau_{E^n}^{x}\}$ has the same eigenvalues and the dependence on $x$ is manifested only in the eigenvectors (See \cref{oper:tau} for a more formal statement).
\item Let $\{\Lambda_{i}\}_{i=0,...,n+1}$ be the set of differing eigenvalues of one block of the operator $\rho_{XE}^{\otimes n}$ in increasing order; i.e. $\Lambda_{i} < \Lambda_{i+1}$ and let $\{m_{i}\}_{i=0,...,n+1}$ be the set of multiplicities such that $m_{i}$ is the multiplicity of $\Lambda_{i}$.  Let $\{\mu_i\}_{i=0,...,n+1}$ be the eigenvalues of one block of  $\tau_{X^n E^n}$ in increasing order with respective multiplicity $\{n_i\}_{i=0,...,n+1}$.
Let
\begin{eqnarray}
	s_r^+&:=&\sum\limits_{i=1}^r m_{n-i+2}(\Lambda_{n-i+2}-
	\Lambda_{n-r+1}),\\
\end{eqnarray}
for $0\le r\le n+1$.

The eigenvalues of $\tau_{X^n E^n}$ are defined  by the following relations 
\begin{equation}
\label{def:eigen:tauxnen}
 \left\{
\begin{array}{l l}
  \mu_{i} := \left\{
  \begin{array}{l l}
    \Lambda_{+} & \quad n+1-b^{+}\leq i\leq n+1\\
    \Lambda_{i} & \quad 1\leq i \leq n-b^{+}\\
	\frac{\myeps}{2m_{0}} & \quad  i=0\\
  \end{array} \right.\\
n_i = m_i  \quad\quad  1\leq i\leq n+1,
\end{array}
 \right.
\end{equation}
where 
\begin{equation}
	b^{+}:=\max\{r:s_r^{+}\le\frac{\myeps}{2}\}
\end{equation}
and
\begin{eqnarray}
\label{eq:s2:lambdaplus}
	\Lambda_{+}:=\Lambda_{n-b^+ + 1}-\frac{\frac{\myeps}{2}-s_{b^+}}{\sum_{i=0}^{b^+}m_{n-i+1}}.
\end{eqnarray}
\end{enumerate}
\end{definition}

Since the smoothing in the smooth R\'enyi entropy of order two is realized by taking the maximum in the environment, it follows for  any operator $\overline{\sigma}_{X^n E^n}\in\mathcal{B}^{\frac{\myeps}{2}}(\rho_{XE}^{\otimes n})$
\begin{equation*}
 S_{2}^{\myeps}(\rho_{XE}^{\otimes n})\geq S_{2}(\overline{\sigma}_{X^n E^n}).
\end{equation*}
Therefore, if we can prove that the operator $\tau_{X^n E^n}$ introduced before is such that $\tau_{X^n E^n}\in\mathcal{B}^{\frac{\myeps}{2}}(\rho_{XE}^{\otimes n})$, then we have proven that the modified smooth R\'enyi entropy is a lower bound for the smooth R\'enyi entropy. 

\begin{proposition}
The operator $\tau_{X^n E^n}$ defined by its eigenvalues in \cref{def:eigen:tauxnen} is such that 
$\frac{1}{2}\|\tau_{X^n E^n} -  \rho_{XE}^{\otimes n}\|=\frac{\myeps}{2}$, i.e., $\tau_{X^n E^n}\in\mathcal{B}^{\frac{\myeps}{2}}(\rho_{XE}^{\otimes n})$.
\end{proposition}
\begin{proof}
The proof follows by the direct calculation of the distance using the spectral decomposition of $\rho_{XE}^{\otimes n}$.
\end{proof}
For the six-state protocol for $n=10^4$,  it turns out that\footnote{The high precision used in this calculation is obtained using an arbitrary precision computer program (See \cref{app:calc}).}  $|\overline{S}_{2}^{\myeps}(\rho_{XE}^{\otimes n})-S_{2}^{\myeps}(\rho_{XE}^{\otimes n})|/S_{2}^{\myeps}(\rho_{XE}^{\otimes n})\propto 10^{-5390}$ for a $QBER=5\%$ and $\myeps=10^{-16}$. Moreover, for increasing $n$ the difference becomes smaller.  The reason of this similarity is that the dimension of the kernel of $\rho_{XE}^{\otimes n}$ is much bigger than the degeneracy of the support, namely  $m_{0}=2^{2n}-2^{n}$ vs $\sum_{i\neq0}m_{i}=2^{n}$, therefore there is, practically,  no difference between the eigenvalue distribution in \cref{def:eigen:tauxnen} and the optimal eigenvalue distribution for $S_{2}^{\myeps}(\rho_{XE}^{\otimes n})$ presented in \cite{tim06}.

\subsubsection{Computable lower bound for the achievable key length\label{sec:depbound}}
The following theorem provides the bound that we are going to exploit in this paper. 

\begin{theorem}
\label{theo:bound:depopt}
Let $\rho_{XE}^{\otimes n}$ be the cq-state describing the classical string shared by Alice and Bob and the correlated quantum state of the  eavesdropper. Then
\begin{align*}
	\underset{\sigma_{X^{n}E^{n}}\in\cqball{\frac{\bar{\myeps}}{2}}{\rho_{XE}^{\otimes n}}}{\mathrm{sup}}&[ S_{2}(\sigma_{X^{n}E^{n}})-S_{0}(\sigma_{E^n})]\geq\\ &\overline{S}_{2}^{\etr}(\rho_{XE}^{\otimes n})-            S^{\erk}_{0}(\rho_{E}^{\otimes n}+\bar{\delta}_{E^{n}}) - \erk
\end{align*} 
with $\bar{\delta}_{E^{n}}=\frac{\erk}{2^{2n+1}}\one_{E^n}$ and $0\leq\erk\leq\overline{\myeps}$.
\end{theorem}

\begin{proof}
In order to provide a lower bound, it is enough to choose an operator in $\mathcal{B}^{\frac{\bar{\myeps}}{2}}_{\mathrm{cq}}(\rho_{XE}^{\otimes n})$ and to calculate the difference between the R\'enyi entropies of the chosen operator. 
In \cref{proof:add:detail} we construct an operator  $\eta_{X^n E^n}\in\mathcal{B}^{\frac{\bar{\myeps}}{2}}_{\mathrm{cq}}(\rho_{XE}^{\otimes n})$ such that the following two inequalities hold:
\begin{equation}\label{proof:prop:s2eq}
S_{2}(\eta_{X^n E^n})\geq \overline{S}_{2}^{\etr}(\rho_{XE}^{\otimes n}) -  \erk
\end{equation}
and
\begin{equation}\label{proof:prop:s0}
S_{0}(\eta_{E^n}) \leq  S_{0}^{\erk}(\rho_{E}^{\otimes n}+\bar{\delta}_{E^{n}}),
\end{equation} 
where $\bar{\delta}_{E^{n}}=\frac{\erk}{2^{2n+1}}\one_{E^n}$.

Using these two inequalities, we have
\begin{align}
	\underset{\sigma_{X^{n}E^{n}}\in \mathcal{B}^{\frac{\bar{\myeps}}{2}  }_{\textrm{cq}}(\rho_{XE}^{\otimes n})}{\mathrm{sup}}&[ S_{2}(\sigma_{X^{n}E^{n}})-S_{0}(\sigma_{E^n})]\geq \\
&S_{2}(\eta_{X^n E^n})-S_{0}(\eta_{E^n})\geq\\
&\overline{S}_{2}^{\etr}(\rho_{XE}^{\otimes n})-            S^{\erk}_{0}(\rho_{E}^{\otimes n}+\bar{\delta}_{E^{n}}) - \erk.
\end{align} 
\end{proof}

\begin{remark}\label{rem:erk} Numerical calculations indicate that the choice $\erk=\frac{\overline{\myeps}}{2}$  is optimal for a wide range of used parameters.
\end{remark}

\begin{remark} The bound provided in \cref{theo:bound:depopt} may not be asymptotically optimal. However, the emphasis is for finite-key analysis and the bound permits to improve the key rate for experimentally relevant number of signals. Note that although we can have small differences in the asymptotic case, the bound is, from the numerical point of view, pretty tight. In fact, note that (see Definition~ \ref{def:smoothrenent})
\begin{align*}
\underset{\sigma_{X^{n}E^{n}}\in\cqball{\frac{\bar{\myeps}}{2}}{\rho_{XE}^{\otimes n}}}{\mathrm{sup}}&[ S_{2}(\sigma_{X^{n}E^{n}})-S_{0}(\sigma_{E^n})]&\\\leq
 &S_{2}^{\overline{\myeps}}(\rho_{XE}^{\otimes n})-S_{0}^{\overline{\myeps}}(\rho_{E}^{\otimes n}).
\end{align*} 

Calculating the difference between the upper bound and the lower bound, it is  for small $n$ ($n\approx10^4$) of the order of $0.1\%$ and it decreases for larger $n$.
\end{remark}

\section{Results\label{sec:result}}
\label{app:calc}

The security is characterized by the parameter $\myeps$, representing the acceptable probability of failure of the execution of the protocol. In the following we consider a standard setting with $\myeps=10^{-9}$.  For the simulations we assume that $n'_{X}=n'_{Y}$ and we take for parameter estimation $m_{X}=m_{Y}=m_{Z}=n'_{X}$. The length of the string used for the extraction of the key is $n=n'_{Z}-m_{Z}$ which has at most QBER $e=e_{m}+2\zeta(\myeps_{PE}, m_{Z})$ with probability $1-\myeps_{PE}$ (see \cref{eq:pa:boundQBER}). The error correction protocol performs such that in \cref{eq:leakec} we have $f_{EC}=1.2$ and  $\myeps_{EC}=10^{-10}$ (\cite{Scarani:2008ys}, \cref{eq:leakec}). Finally, we optimize the free parameters ($\myeps_{PE}, \epsbar, \myeps_{PA}, \mathrm{Pr}(\sigma_{X})$,$n$) in order to  maximize the key rate. 

The algorithms for the calculations were implemented using \verb!C++!. The library  CNL (Class Library for Numbers) \cite{web:cln} was used to perform calculations with arbitrary precision. Due to the non-smoothness of the involved functions, we used the Hybrid Optimization Parallel Search PACKage HOPSPACK \cite{Hops20-Sandia}, which permits to deal with all involved optimizations in an efficient way and permits to perform the calculations on a cluster.

In the following we summarize the three bounds for the achievable secret key rate that we are going to compare.

\subsubsection*{Bound proposed in this paper}
The following proposition summarizes our results of \cref{Sec:security}.
\begin{proposition}
Let $\rho_{XE}^{\otimes n}$ be the cq-state describing the classical string shared by Alice and Bob which is correlated with the quantum state of the  eavesdropper. Let $N$ be the initial number of quantum states shared by Alice and Bob, $n$ be the length of the string used for extracting the key which has QBER $e=e_{m}+2\zeta(\myeps_{PE}, m_{Z})$ with probability $1-\myeps_{PE}$. Then Alice and Bob can achieve the secret key rate
\begin{align}
\label{eq:r}
        r:=&\frac{n}{N}\left[\overline{S}_{2}^{\frac{\bar{\myeps}}{2}}(\rho_{XE}^{\otimes n})-            S^{\frac{\bar{\myeps}}{2}  }_{0}(\rho_{E}^{\otimes n}+\bar{\delta}_{E^{n}}) - \overline{\myeps}  -\leak \right]_{e=e_{m}+2\zeta}\nonumber\\ &\quad\quad\quad+2\log\left(2\myeps_{PA}\right),
        \end{align}
where $\myeps=\myeps_{PE}+\myeps_{PA}+\overline{\myeps}+\myeps_{EC}$.
\end{proposition}
\begin{proof}
Using \cref{theorem:PADEP}, \cref{theo:bound:depopt} and Remark~\ref{rem:erk} the result follows.
\end{proof}

\subsubsection*{Asymptotic Equipartition Property (AEP) bound}
The conditional smooth min-entropy \cite{Renner:kl} characterizes the optimal secret key rate \cite{Renner:kl, Tomamichel2010}. The AEP bound  used in \cite{Scarani:2008ys,bratzik2010min}  comes from the AEP approximation \cite{Renner:kl, Tomamichel:2009cr} of the conditional smooth min-entropy.
Collective attacks allow us to bound the smooth min-entropy of a product state by the conditional von Neumann entropy of a single state \cite{bratzik2010min,  Scarani:2008ve}
  \begin{align}
\label{eq:raep}
  r_{\mathrm{AEP}}:=&\frac{n}{N}\left[ H(X|E)_{\rho}-5\sqrt{\frac{\log(2/\epsbar)}{n}}-\leak \right]_{e=e_{m}+2\zeta}\nonumber\\ &\quad\quad\quad+2\log\left(2\myeps_{PA}\right),
  \end{align} with $H(X|E)_{\rho}=(1-e)\left[1-h\left(\frac{1-\frac{3}{2}e}{1-e}\right)\right]$.
  
\subsubsection*{Smooth R\'enyi entropy bound}This bound was derived in \cite{renner05} and calculated in \cite{tim06} and is given by
\begin{align}
\label{eq:rsre}
        r_{\mathrm{SRE}}:=&\frac{n}{N}\left[\ren{2}{\eps{}'}{\rho_{XE}^{\otimes n}}-\ren{0}{\eps{}'}{\rho_{E}^{\otimes n}}-\leak \right]_{e=e_{m}+2\zeta}\nonumber\\ &\quad\quad\quad+2\log\left(2\myeps_{PA}\right),
        \end{align}
        where $\myeps'=\frac{\overline{\myeps}^2}{2}$.

\subsection{Robustness of the protocol}
An important figure of merit is the \emph{threshold QBER} which characterizes the minimal N for a fixed QBER permitting to extract a positive key rate. As shown in Fig.~\ref{fig:threshold}, with the bound presented in this paper it is possible to have a positive key rate with 23\% signals less than the smooth R\'enyi entropy bound  and 50\% signals less than the AEP approach, for a QBER of 1\%.
% For other values of the QBER, we obtain only small changes in the relative ratio. For example, for a QBER of $9\%$ the advantage of our approach on the smooth R\'enyi entropy bound is of 20\% and on the AEP approach is of 48\%.

\begin{figure}[h]
\includegraphics[width=8.6cm]{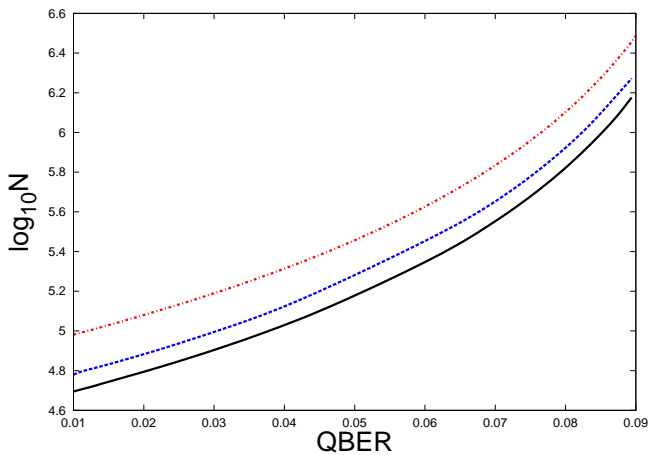}
\caption{\label{fig:threshold}(Color online) Minimal number of signals versus QBER permitting to extract a non-zero key rate. Comparison between the bound presented in this article $r$ (black solid line), see \cref{eq:r}, the smooth R\'enyi entropy bound $r_{SRE}$ (blue dashed line), see \cref{eq:rsre} and the AEP bound $r_{AEP}$ (red dot-dashed  line), see \cref{eq:raep}.}
\end{figure}

\subsection{Secret key rates}
In  \cref{fig:kr1}, we compare the secret key rates calculated by the three approaches for various QBERs. The bound developed in this paper leads to significant higher key rates  when limited resources are used. In particular when $QBER=1\%$ with the bound presented in this paper with $N\approx5\cdot10^4$, it is possible to have non-zero key rates. Instead with the other approaches it is necessary to use $N\approx6.5\cdot10^4$ for the smooth R\'enyi entropy bound and $N\approx10^5$ for the AEP bound.

\begin{figure}[h]
\includegraphics[width=8.6cm]{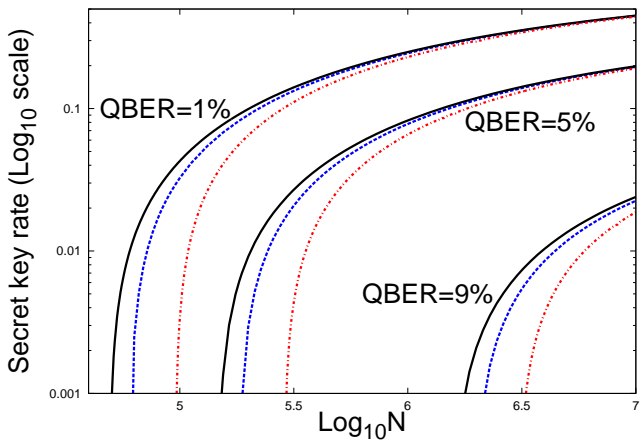}
\caption{\label{fig:kr1}(Color online) Key rate versus $\log_{10}(N)$ where N is the initial number of quantum systems shared by Alice and Bob.  Comparison between the bound presented in this article $r$ (black solid line), see \cref{eq:r}, the smooth R\'enyi entropy bound $r_{SRE}$ (blue dashed line), see \cref{eq:rsre} and the AEP bound $r_{AEP}$ (red dot-dashed  line), see \cref{eq:raep}. }
\end{figure}

\section{Conclusions\label{sec:conclusion}}
Although optimal bounds for the finite-key scenario are provided in the literature they are not calculable and were so far only estimated by  bounds coming from the asymptotic equipartition theorem (see \cite{bratzik2010min, TomLim2011} for two exceptions). In this paper we resumed the smooth R\'enyi entropy bound \cite{tim06} and we proved that this bound is tighter than the AEP bound. Our main contribution is a new bound on the maximal achievable secret key length which involves optimizations on R\'enyi entropies. With respect to \cite{tim06} the main advantage is that we use a bigger environment for the optimizations and with respect  to \cite{Scarani:2008ve, Scarani:2008ys} we don't use bounds coming from corrections to the asymptotic case.  As a result we were able to obtain  higher secret key rates with respect to \cite{tim06,Scarani:2008ve, Scarani:2008ys}. For calculating the quantities involved in our analysis we need  the quantum channel to be symmetric. Although we don't have any guarantee that Alice and Bob share such a channel, it is possible for them to reduce to this case employing additional symmetries\footnote{Actually, in this case it also possible to redefine the protocol removing the sifting following the construction presented in \cite{kraus:2005kx}. The key rate will be higher but the relative differences between the three approaches remain the same.} or taking as QBER of a symmetric channel the worst one of a non-symmetric channel.

Finally, regarding future work, note that here we considered an ideal protocol where the signals entering in Alice and Bob's laboratory are qubits and where the measurement devices are perfect. All this assumptions could be relaxed following the analysis done in \cite{Scarani:2009wr, scarani:2009}.

\begin{acknowledgments}
We would like to thank Sylvia Bratzik, Matthias Kleinmann, and in particular Renato Renner for valuable and enlightening discussions. AS thanks also Alberto Carlini for his interest, advice and support during the early stage of this work. We acknowledge partial financial support by Deutsche Forschungsgemeinschaft (DFG) and by BMBF (project QuOReP).
\end{acknowledgments}

\begin{appendix}

\section{Properties of R\'{e}nyi entropies}\label{app:prop:renyi}
The following properties and their proofs can be found in \cite{Renner:2005pi} and \cite{RenWol04a}.
 \begin{lemma}[Data processing] \label{lem:KRDATA}
 Let $\eps,\eps'\geq 0$ and $\rho_{XBC}\in S\left(\mathcal{H}^X \otimes \mathcal{H}^B \otimes \mathcal{H}^C \right)$ be a classical-quantum state, i.e. $\rho_{XBC}=\sum_{x\in\mathcal{X}}P_X(x)\pr{x}\otimes \rho_B^x\otimes \rho_C^x$. Then with $\ren{2}{}{\rho_{XC}|X}:=\inf_{x \in \mathcal{X}}\ren{2}{}{\rho_{XC}|x}$, the following inequality holds
  \begin{equation}
  \ren{2}{\eps{}+\eps{}'}{\rho_{XBC}}\geq \ren{2}{\eps{}'}{\rho_{XC}}+\ren{2}{\eps{}}{\rho_{XB}|X}.
  \end{equation}
 \end{lemma}
 \begin{lemma}[Subadditivity]\label{lem:KRSUBADD}
 Let $\eps{}\geq0,\eps{}'\geq 0$ and $\rhoAB\in S\left(\mathcal{H}^A \otimes \mathcal{H}^B \right)$, then
  \begin{equation}
  \ren{0}{\eps{}+\eps{}'}{\rhoAB}\geq \ren{0}{\eps{}}{\rho_{A}}+\ren{0}{\eps{}'}{\rho_{B}}.
  \end{equation}
 \end{lemma}

\section{Proof of \cref{theorem:PADEP}} \label{proof:main:theo}
Before we start with the proof, we define some quantities used in the following.  

 \begin{definition}(Two-universal hash functions\cite{Carter1979143})\label{defi:hash}
  Let $\mathcal{F}$ be a family of functions from $\mathcal{X}$ to $\mathcal{Z}$ and let $P_{\mathrm{F}}$ be a probability distribution on $\mathcal{F}$. The pair $\left(\mathcal{F},P_{\mathrm{F}}\right)$ is called \cal{two-universal} if $\mathrm{P}_f\left[f(x)=f(x')\right]\leq\frac{1}{|\mathcal{Z}|}$ for any distinct $x,x'\in \mathcal{X}$ and $f$ chosen at random from $\mathcal{F}$ according to the distribution $P_{\mathrm{F}}$.
 \end{definition}

The following definition involves  the leakage of information during the error correction protocol.
\begin{definition}\label{defi:ECLEAK}
The number of bits leaked to an eavesdropper during the error correction protocol is \cite{Renner:kl}
  \begin{equation}
   \leak:=\log{|\mathcal{C}|}-\inf_{x^n \in \mathcal{X}}\ren{\infty}{}{P_{C|X^n=x^n}},
   \end{equation}
  where $|\mathcal{C}|$ is the cardinality of the set $\mathcal{C}$ containing all possible communication transcripts and $P_{C|X^n=x^n}$ is the probability that there is a specific communication transcript when Alice has a specific $x^n$. 
  \end{definition}
Note that in the definition Bob is missing because we consider a one-way error correction protocol.

Moreover, let us recall a result proven in \cite{renner05} and used in the following proof.

\begin{theorem}\cite{renner05} Let $\rho_{X^nE^nC}$ be a cq-state describing Alice's bitstring $X^n$, Eve's quantum system and the distribution of error correction transcripts $C$. Let $\mathcal{F}$ be a two-universal family of hash function from $\mathcal{X}^{n}\rightarrow\{0,1\}^{\ell}$. Then 
 \begin{align}\label{eq:PROOFENV}
 \frac{1}{2}\|\rho_{F(X^n)^\ell E^lCF}-&\rho_U\otimes\rho_{E^\ell CF}\|_{1}\nonumber\\
  &\leq \frac{1}{2}2^{-\frac{1}{2}\left(\ren{2}{}{\rho_{X^nE^nC}}-\ren{0}{}{\rho_{E^nC}}-\ell\right)}
 \end{align} where $\rho_{F(X^n)^\ell E^\ell CF}:=\sum_{f\in \mathcal{F}}P_{\mathrm{F}}(f) \rho_{f(X^n)^\ell E^\ell C}\otimes \pr{f}$ and $\rho_U=\frac{1}{2^\ell}\one$.
\end{theorem}

Now we are ready to prove \cref{theorem:PADEP}.
 \begin{proof}(\Cref{theorem:PADEP}) 
  At the end of the QKD protocol, the classical string obtained from privacy amplification correlated with Eve's information is 
 \begin{equation*}
  \rho_{F(X^n)^\ell E^\ell CF}:=\sum_{f\in \mathcal{F}}P_{\mathrm{F}}(f) \rho_{f(X^n)^\ell E^\ell C}\otimes \pr{f}.
 \end{equation*}
   Let $\rho'_{F(X^n)^\ell E^\ell CF} \in \ball{\frac{\epsbar}{2}}{\rho_{F(X^n)^\ell E^\ell CF}}$ be the operator that maximizes the right-hand side of  \cref{eq:PROOFENV}. Because the trace-distance does not increase applying the partial trace, it follows that $\rho'_{E^\ell CF}\in \ball{\frac{\epsbar}{2}}{\rho_{E^\ell CF}}$. Let us define  $\rho_U=\frac{1}{2^\ell}\one$. Then, 
  \begin{eqnarray}
  &&  \dist{\rho_{F(X^n)^\ell E^\ell CF}-\rho_U\otimes\rho_{E^\ell CF}}\nonumber \\
  &=& \frac{1}{2} \|\rho_{F(X^n)^\ell E^\ell CF}-\rho'_{F(X^n)^\ell E^\ell CF}\nonumber\\ &\quad&\quad\quad\quad+\rho'_{F(X^n)^\ell E^\ell CF}-\rho_U\otimes \rho'_{E^\ell CF}\nonumber\\ &\quad&\quad\quad\quad-\rho_U\otimes\rho_{E^\ell CF}+\rho_U\otimes\rho'_{E^\ell CF}\|_{1} \label{proof:depbound:2} \\
  &\leq& 2\frac{\epsbar}{2}+\dist{\rho'_{F(X^n)^\ell E^\ell CF}-\rho_U\otimes \rho'_{E^\ell CF}} \label{proof:depbound:3} \\
  &\leq& \epsbar+\frac{1}{2}2^{-\frac{1}{2}\left(\sup_{\sigma_{X^nE^nC}\in   \ball{\frac{\overline{\eps}}{2}}{\rho_{X^nE^nC}}}\left[\ren{2}{}{\sigma_{X^nE^nC}}-\ren{0}{}{\sigma_{E^nC}}\right]-\ell\right) }\nonumber.
  \end{eqnarray} 
  In the step from Eq.~\eqref{proof:depbound:2} to Eq.~\eqref{proof:depbound:3} we used the triangle inequality and the fact that the maximal possible distance is $\frac{\overline{\myeps}}{2}$. The last inequality follows from \cref{eq:PROOFENV} and the definition of ${\rho'_{F(X^n)^\ell E^\ell CF}}$.
Requiring that the distilled key is  $\left(\epsbar+\myeps_{PA}\right)$-secure, i.e.
  \begin{align*}
  \epsbar+\frac{1}{2}&2^{-\frac{1}{2}\left(\sup_{\sigma_{X^nE^nC}\in \ball{\frac{\overline{\eps}{}}{2}}{\rho_{X^nE^nC}}}\left[\ren{2}{}{\sigma_{X^nE^nC}}-\ren{0}{}{\sigma_{E^nC}}\right]-\ell\right) }\\ &\stackrel{\textbf{!}}{\leq} \epsbar +\myeps_{PA} ,
  \end{align*}
 the proof of \cref{proof:step1} is completed. 
  
  Regarding the leakage term, note that in order to apply Lemma~\ref{lem:KRDATA} of \cref{app:prop:renyi} for bounding $\ren{2}{}{\sigma_{X^nE^nC}}$  we restrict the optimization  space to $ \cqball{\frac{\overline{\eps}}{2}}{\rho_{X^nE^nC}}$. Therefore, 
 \begin{align*}
  & \sup_{\sigma_{X^nE^nC}\in \ball{\frac{\overline{\eps}}{2}}{\rho_{X^nE^nC}}}{\left(\ren{2}{}{\sigma_{X^nE^nC}}-\ren{0}{}{\sigma_{E^nC}}\right)}\geq \\
  &\sup_{\sigma_{X^nE^nC}\in \cqball{\frac{\overline{\eps}}{2}}{\rho_{X^nE^nC}}}{\left(\ren{2}{}{\sigma_{X^nE^nC}}-\ren{0}{}{\sigma_{E^nC}}\right)}.
  \end{align*}
  Using Lemma~\ref{lem:KRDATA} of \cref{app:prop:renyi}  with $\myeps=\myeps'=0$, it follows that
  \begin{align}
  \ren{2}{}{\sigma_{X^nE^nC}} \geq  \ren{2}{}{\sigma_{X^nE^n}} + S_{2}(\sigma_{X^nC}|X^{n}).
  \end{align}
  By definition
  \begin{align}
    S_{2}(\sigma_{X^nC}|X^{n})&:=\inf_{x^n \in \mathcal{X}}\ren{2}{}{P_{C|X^n=x^n}}\\
    &\geq \inf_{x^n \in \mathcal{X}}\ren{\infty}{}{P_{C|X^n=x^n}}.
   \end{align}
  Moreover, using Lemma~\ref{lem:KRSUBADD} with $\myeps=\myeps'=0$ we obtain
  \begin{equation}
  \ren{0}{}{\sigma_{E^nC}}\geq \ren{0}{}{\sigma_{E^n}} + \ren{0}{}{\sigma_{C}}.
  \end{equation}
  Putting together the last four equations and Definition~\ref{defi:ECLEAK} the proof is concluded.
\end{proof}

\section{The operator $\rho_{XE}$}
\label{sec:usedoper}
\label{sec:rhoxe}
The Bell-diagonal state shared by Alice and Bob after the use of the depolarizing map is
\begin{align*}
\rhoAB&=\lambda_{1}\ket{\psi^{+}}\bra{\psi^{+}}+\lambda_{2}\ket{\psi^{-}}\bra{\psi^{-}}\\
&+\lambda_{3}\ket{\phi^{+}}\bra{\phi^{+}}+\lambda_{4}\ket{\phi^{-}}\bra{\phi^{-}},
\end{align*}
where the states $\{\ket{\psi^{\pm}}, \ket{\phi^{\pm}}\}$ are the Bell states and $\sum_{i}\lambda_{i}=1$. 

For the symmetric six-state protocol
\begin{equation}
\label{def:lambdas}
\lambda_{0}=\frac{1}{2}(2-3e), \quad \lambda_1=\lambda_2=\lambda_3=\frac{e}{2},
\end{equation}
where  $e$ is the Quantum Bit Error Rate (QBER).

The operator $\rho_{ABE}$  is defined as the purification of $\rhoAB$. Tracing out Bob and measuring Alice's system, we get the operator $\rho_{XE}^{\otimes n}$ describing the classical string $X^n$ held by Alice and Bob and Eve's quantum systems $\rho^{\otimes n}_E$. In general\begin{equation}
\label{def:rhoxe}
\rho_{XE}^{\otimes n}:=\left(\rho_{E}^{0}\oplus\rho_{E}^{1}\right)^{\otimes n},
\end{equation}
with
\begin{align}
  \rho_E^0&:=\left(
    \begin{array}{cccc}
      \lambda _0 &  \sqrt{\lambda _0 \lambda _1} & 0 & 0  \\
      \sqrt{\lambda _0 \lambda _1} & \lambda _1 & 0 & 0  \\
      0 & 0 & \lambda _2 &  \sqrt{\lambda _2 \lambda _3} \\
      0 & 0 & \sqrt{\lambda _2 \lambda _3} & \lambda _3 
    \end{array}
  \right),\\
   \rho_E^1&:=\left(
    \begin{array}{cccc}
      \lambda _0 &  -\sqrt{\lambda _0 \lambda _1} & 0 & 0  \\
      -\sqrt{\lambda _0 \lambda _1} & \lambda _1 & 0 & 0  \\
      0 & 0 & \lambda _2 &  -\sqrt{\lambda _2 \lambda _3} \\
      0 & 0 & -\sqrt{\lambda _2 \lambda _3} & \lambda _3 
    \end{array}
  \right).
\end{align}

Diagonalizing the operators above, we find that they have the same eigenvalues but different eigenvectors, i.e.
\begin{equation}
\rho_{E}^{x}:=\sum_{i=0}^{3}\Gamma^{(1)}_{i}P_{i}^{x}\label{eq:rhoedirac2},
\end{equation}
where the eigenvalues $\{\Gamma_{i}^{(1)}\}$ are
\begin{subequations}
\label{eigenrxeall}
\begin{align}
 \Gamma^{(1)}_{0}&=\Gamma^{(1)}_{2}=0  \label{eigenrxeker},\\
 \Gamma^{(1)}_{1}&=\lambda _0+\lambda _1 \label{eigenrxesup1},\\
 \Gamma^{(1)}_{3}&=\lambda _2+\lambda _3, \label{eigenrxesup2}
\end{align}  
\end{subequations}
and the operators $\{P_{i}^{x}\}$ are projectors on the eigenspace of the eigenvalues $\{\Gamma_{i}^{(1)}\}$ obtained by diagonalizing $\rho_{E}^{x}$. From the diagonalization it is possible to derive explicitly the projectors: \footnote{Let us define $\mathbb{O}_{2}$ as the 2x2 matrix with zero entries.}

\begin{subequations}\label{proj:p}
\begin{align*}
\scriptstyle P_{0}^{0}& \scriptstyle = \scriptstyle\left(\begin{array}{cc}
 \frac{\lambda _1}{\lambda _0+\lambda _1} & -\frac{\sqrt{\lambda
   _0 \lambda _1}}{\lambda _0+\lambda _1} \\
 -\frac{\sqrt{\lambda _0 \lambda _1}}{\lambda _0+\lambda _1} &
   \frac{\lambda _0}{\lambda _0+\lambda _1}
\end{array}
\right)\oplus \mathbb{O}_{2}, \scriptstyle P_{1}^{0}=\left(
\begin{array}{cc}
 \frac{\lambda _0}{\lambda _0+\lambda _1} & \frac{\sqrt{\lambda _0
   \lambda _1}}{\lambda _0+\lambda _1} \\ 
 \frac{\sqrt{\lambda _0 \lambda _1}}{\lambda _0+\lambda _1} &
   \frac{\lambda _1}{\lambda _0+\lambda _1}
\end{array}
\right)\oplus \mathbb{O}_{2}  \\ \scriptstyle P_{2}^{0}& \scriptstyle = \scriptstyle\mathbb{O}_{2}\oplus \left(
\begin{array}{cc}
 \frac{\lambda _3}{\lambda _2+\lambda _3} & -\frac{\sqrt{\lambda
   _2 \lambda _3}}{\lambda _2+\lambda _3} \\
 -\frac{\sqrt{\lambda _2 \lambda _3}}{\lambda _2+\lambda _3} &
   \frac{\lambda _2}{\lambda _2+\lambda _3}
\end{array}
\right), \scriptstyle P_{3}^{0}=\mathbb{O}_{2}\oplus\left(
\begin{array}{cc}
 \frac{\lambda _2}{\lambda _2+\lambda _3} & \frac{\sqrt{\lambda _2
   \lambda _3}}{\lambda _2+\lambda _3} \\
 \frac{\sqrt{\lambda _2 \lambda _3}}{\lambda _2+\lambda _3} &
   \frac{\lambda _3}{\lambda _2+\lambda _3}
\end{array}
\right)\\
\scriptstyle P_{0}^{1}& \scriptstyle = \left(
\begin{array}{cc}
 \frac{\lambda _1}{\lambda _0+\lambda _1} & \frac{\sqrt{\lambda _0
   \lambda _1}}{\lambda _0+\lambda _1} \\
 \frac{\sqrt{\lambda _0 \lambda _1}}{\lambda _0+\lambda _1} &
   \frac{\lambda _0}{\lambda _0+\lambda _1}
\end{array}
\right) \scriptstyle\oplus \mathbb{O}_{2}, \scriptstyle P_{1}^{1}=\left(
\begin{array}{cc}
 \frac{\lambda _0}{\lambda _0+\lambda _1} & -\frac{\sqrt{\lambda
   _0 \lambda _1}}{\lambda _0+\lambda _1} \\
 -\frac{\sqrt{\lambda _0 \lambda _1}}{\lambda _0+\lambda _1} &
   \frac{\lambda _1}{\lambda _0+\lambda _1}
\end{array}
\right)\oplus \scriptstyle \mathbb{O}_{2}\\ \scriptstyle P_{2}^{1}& \scriptstyle = \scriptstyle \mathbb{O}_{2}\oplus\left(
\begin{array}{cc}
 \frac{\lambda _3}{\lambda _2+\lambda _3} & \frac{\sqrt{\lambda _2
   \lambda _3}}{\lambda _2+\lambda _3} \\
 \frac{\sqrt{\lambda _2 \lambda _3}}{\lambda _2+\lambda _3} &
   \frac{\lambda _2}{\lambda _2+\lambda _3}
\end{array}
\right), \scriptstyle P_{3}^{1}= \scriptstyle \mathbb{O}_{2}\oplus\left(
\begin{array}{cc}
 \frac{\lambda _2}{\lambda _2+\lambda _3} & -\frac{\sqrt{\lambda
   _2 \lambda _3}}{\lambda _2+\lambda _3} \\
 -\frac{\sqrt{\lambda _2 \lambda _3}}{\lambda _2+\lambda _3} &
   \frac{\lambda _3}{\lambda _2+\lambda _3}
\end{array}\right).
\end{align*}
\end{subequations}

For the following, it is convenient to define 
\begin{equation}
  \label{proj:sum}
  P_{i}:=\frac{1}{2}\sum_{x=0}^{1}P_{i}^{x},
 \end{equation}
 with $i=0,1,2,3$. As can be easily verified the operators $\{P_{i}\}$ are diagonal in the basis where the operators $\{P^x_i\}$ assume the form given above.
 
\section{Additional details of the proof of \cref{theo:bound:depopt}}
\label{proof:add:detail}
Before to start, it is necessary to fix the notation for the involved operators.
The operator $\rho_{XE}^{\otimes n}$ can be written as 
\begin{equation}
\label{eq:rhoxeformal}
\rho_{XE}^{\otimes n}=\frac{1}{2^{n}}\sum_{x=0}^{2^{n}-1}\ket{x}\bra{x}\otimes\sum_{i=0}^{2^{2n}-1}\Gamma_{i}^{(n)}P_{i}^{(n)x},
\end{equation}
where 
\begin{align}
 \label{proj:nix} P_{i}^{(n)x}&:=\bigotimes_{p=0}^{n-1}P_{i_{p}}^{x_{p}},  \\
\label{prod:eigenv} \Gamma_{i}^{(n)}&:=\prod_{p=0}^{n-1}\Gamma_{i_{p}}^{(1)}
\end{align}
and $i:=\sum_{p=0}^{n-1}4^{p}i_{p}$ with $i_{p}=0,...,3$, $x:=\sum_{p=0}^{n-1}2^{p}x_{p}$ where $x_{p}$ is a binary digit.

The operator $\rho_{XE}^{\otimes n}$ is constituted of $2^n$ diagonal blocks labeled by the index $x$. Each of the blocks has $2^{2n}$ eigenvalues $\Gamma^{(n)}_{i}$ and each eigenvalue is associated to a projector $P_{i}^{(n)x}$ that depends on the eigenvalue (index "i") and on the block (index "x").

\subsection{Construction of the operator $\eta_{X^nE^n}$}
The operator $\eta_{X^nE^n}$ is constructed in such a way that the inequalities in \cref{proof:prop:s2eq} and \cref{proof:prop:s0} are satisfied. For constructing $\eta_{X^nE^n}$ we construct first $\eta_{E^n}$  using the following two steps:
\begin{enumerate}
\item Find $\tau_{X^n E^n}$ such that $S_{2}(\tau_{X^n E^n})=\overline{S}_{2}^{\etr}(\rho_{XE}^{\otimes n})$
\item Find $\eta_{X^n E^n}$ such that $S_0(\eta_{E^n})=S_{0}^{\erk}(\tau_{E^n})$.
\end{enumerate}
By definition of smooth R\'enyi entropy of order two $\tau_{X^n E^n}\in \cqball{\frac{\etr}{2}}{\rho_{XE}^{\otimes n}}$  and it can be written as

\begin{equation}
\label{oper:tau}
\tau_{X^n E^n}:=\frac{1}{2^{n}}\sum_{x=0}^{2^{n}-1}\ket{x}\bra{x}\otimes\sum_{i=0}^{2^{2n}-1}\tau_{i}^{(n)}P_{i}^{(n)x}.
\end{equation}

 Regarding the operator $\eta_{E^n}$, the constraint on its R\'enyi entropy of order zero is only a constraint on its eigenvalues.  For assigning a well defined structure of operator to $\eta_{X^nE^n}$  we use $\tau_{X^nE^n}$. Let $\soproj$ be the projector that cuts out the eigenvalues of $\tau_{E^n}$ such that their sum is $\erk/2$. 

The operator $\eta_{X^nE^n}$, is defined by
\begin{equation}
\label{def:eta}
 \eta_{X^n E^n}:=\frac{1}{1-\frac{\erk}{2}}(\one_{X^n}\otimes \soproj)\tau_{X^{n}E^{n}}(\one_{X^n}\otimes \soproj ).
 \end{equation} 

This definition is such that $\eta_{E^n}=\frac{1}{1-\frac{\erk}{2}}\soproj\tau_{E^n}\soproj$ has the eigenvalues for respecting $S_0(\eta_{E^n})=S_{0}^{\erk}(\tau_{E^n})$.

Finally, note, that the construction above, although arbitrary, is legitimate because, as it is easy to verify, $\eta_{X^n E^n}\in \cqball{\frac{\overline{\myeps}}{2}}{\rho_{XE}^{\otimes n}}$ as required by the statement of \cref{theo:bound:depopt}.

\subsection{Proof of $S_{0}(\eta_{E^n}) \leq  S_{0}^{\erk}(\rho_{E}^{\otimes n}+\bar{\delta}_{E^{n}})$}
\label{sec:upps0}

In order to find the claimed bound, we need to find a bound on the eigenvalues of the operator $\tau_{E^n}$. In order to do that, we exploit the definition of $\tau_{X^n E^n}$, $\rho_{XE}^{\otimes n}$ and of modified smooth R\'enyi entropy of order two (Defintion~\ref{def:modsmoothrentwo}).

We introduce the operator $\delta_{X^{n}E^{n}}$ defined by
\begin{equation}\label{eq:def:delta}
\delta_{X^{n}E^{n}}:=\tau_{X^n E^n}-\rho_{XE}^{\otimes n}.
\end{equation}

Let $\delta_{E^n}:=\tr_{E^n}(\delta_{X^nE^n})$, $\tau_{E^n}:=\tr_{E^n}(\tau_{X^nE^n})$ and let $\{\ket{l}\}$ be a basis of eigenvectors of the operator $\tau_{E^n}$. The eigenvalues  of $\tau_{E^n}$ are
\begin{align}
\bra{l}\tau_{E^n}\ket{l}&:=\bra{l}\rho_{E}^{\otimes n}\ket{l}+\bra{l}\delta_{E^n}\ket{l}.
\end{align}

The operator $\rho_{E}^{\otimes n}$ is fully characterized by the protocol \cite{tim06}. In order to complete the proof, it remains to bound $\bra{l}\delta_{E^n}\ket{l}$.

The following lemma permits to reduce the analysis to the eigenvalues of $\delta_{E^n}$.

\begin{lemma}
\label{theo:comm}
\label{prop:commut}
Let $\rho_{XE}^{\otimes n}, \tau_{X^n E^n}$ be the operators described by equation \eqref{eq:rhoxeformal} and \eqref{oper:tau}. Then
 \begin{equation*}
  [\tau_{E^{n}},\rho_{E}^{\otimes n}]=0.
 \end{equation*}
\end{lemma}
\begin{proof} By definition 
\begin{align*}
 \tau_{E^{n}}&=\frac{1}{2^{n}}\sum_{x=0}^{2^{n}-1}\sum_{i=0}^{2^{2n}-1}\tau_{i}^{(n)}P_{i}^{(n)x}\nonumber
 &=\sum_{i=0}^{2^{2n}-1}\tau_{i}^{(n)}\left(\frac{1}{2^{n}}\sum_{x=0}^{2^{n}-1}P_{i}^{(n)x}\right).
\end{align*}
Observe that the operator in the brackets is diagonal, in fact:
\begin{align*}
\label{proj:decom}
 \frac{1}{2^{n}}\sum_{x=0}^{2^{n}-1}P_{i}^{(n)x}&=\frac{1}{2^{n}}\sum_{x=0}^{2^{n}-1}\bigotimes_{p=0}^{n-1}P_{i_{p}}^{x_{p}}\\
 &=\frac{1}{2^{n}}\sum_{x_{0}=0}^{1}\sum_{x_{1}=0}^{1}...\sum_{x_{n-1}=0}^{1}\bigotimes_{p=0}^{n-1}P_{i_{p}}^{x_{p}}\\
 &=\bigotimes_{p=0}^{n-1}\left(\frac{1}{2}\sum_{x=0}^{1}P_{i_{p}}^{x}\right)\\
 &=\bigotimes_{p=0}^{n-1}P_{i_{p}}.
\end{align*}
Due to the diagonality of the operators  $\{P_{i_{p}}\}_{i_{p}=0,...,3}$ and the fact that the tensor product of diagonal operators lead to a diagonal operator, the assertion is proved. 
\end{proof}

The next lemma, permits to relate the operator $\delta_{E}$ to the operators $P_{i}$ defined in \cref{proj:sum}.

\begin{lemma} It holds
\begin{equation*}
 \delta_{E^{n}}=\frac{\etr}{2m_{0}}\left(\one-(P_{1}+P_{3})^{\otimes n}\right)+\sum_{i\in \mathcal{V}}\delta_{i}P_{i}^{(n)},
\end{equation*}
where the operators $P_{i}$ are defined in \cref{proj:sum}, $P_{i}^{(n)}=\bigotimes_{p=0}^{n-1}P_{i_{p}}$  and $\mathcal{V}:=\{i : \Gamma_{i}^{(n)}\neq 0\}$.
\end{lemma}

\begin{proof}
Using the eigenvalues in \cref{def:eigen:tauxnen} and  \cref{eq:def:delta}, \cref{oper:tau}
\begin{align*}
	\delta_{E^{n}}&:=\tr_{X^n}(\delta_{X^n E^n})\\
&=\frac{\etr}{2m_{0}}\frac{1}{2^{n}}\sum_{x=0}^{2^{n}-1}\sum_{i\in\mathcal{V}^{\bot}}P_{i}^{(n)x} \\ &\quad\quad+\frac{1}{2^{n}}\sum_{x=0}^{2^{n}-1}\sum_{i\in\mathcal{V}}\delta_{i}P_{i}^{(n)x}.
\end{align*}

The quantity $\sum_{i\in\mathcal{V}^{\bot}}P_{i}^{(n)x}$ is such that
\begin{equation*}
 \sum_{i\in\mathcal{V}^{\bot}}P_{i}^{(n)x}=\one-\sum_{i\in\mathcal{V}}P_{i}^{(n)x}.
\end{equation*}
From equation \cref{prod:eigenv} and \cref{eigenrxeall}, we see that the non-zero eigenvalues of $
\rho_{XE}^{\otimes n}$  are characterized by the absence of the index $i_{p}=0,2$. Therefore, it follows that 
\begin{align*}
 \sum_{i\in \mathcal{V}}P_{i}^{(n)x}&\overset{\eqref{proj:nix}}{=}\sum_{i\in  \mathcal{V}}\bigotimes_{p=0}^{n-1} P_{i_{p}}^{x_{p}}\\
 &=\bigotimes_{p=0}^{n-1} \left(P_{1}^{x_{p}}+P_{3}^{x_{p}}\right).
\end{align*}
By taking  the sum over all blocks and using the operator defined in \cref{proj:sum} the statement of the lemma follows.
\end{proof}

Using the previous lemma we can prove the most important result of this section.
\begin{proposition}
\label{prop:max}
 For $\lambda_{0}>\frac{1}{2}$ and $\lambda_{1}=\lambda_{2}=\lambda_{3}$ the following inequality holds:
 \begin{equation}
  \bra{l}\delta_{E^{n}}\ket{l}\leq\frac{\etr}{2m_{0}}\left(1-\left(\frac{\lambda_{1}}{\lambda_{1}+\lambda_{0}}\right)^{n}\right),
 \end{equation}
where $\{\ket{l}\}$ is a basis of eigenvectors for the operator $\rho_{E}^{\otimes n}$.
\end{proposition}

\begin{proof}
\begin{align}
 \bra{l}\delta_{E^{n}}\ket{l}&\leq\frac{\etr}{2m_{0}} \underset{\{\ket{l}\}}{\mathrm{max}} \bra{l}\left(\one-(P_{1}+P_{3})^{\otimes n}\right)\ket{l} + \nonumber\\ &\quad\quad\quad  +\sum_{i\in\mathcal{V}} \delta_i \bra{l}P^{(n)}_{i}\ket{l}\nonumber\\
 &=\frac{\etr}{2m_{0}}\left(1-\left(\frac{\lambda_{1}}{\lambda_{1}+\lambda_{0}}\right)^{n}\right) + \nonumber\\ &\quad\quad\quad +\sum_{i\in\mathcal{V}} \delta_i \bra{l}P^{(n)}_{i}\ket{l}, \label{proof:prop:max:eq} 
\end{align}
where  $\mathcal{V}:=\{i : \Gamma_{i}^{(n)}\neq 0\}$.
Since the $\delta_{i}$  are  negative or zero and the operators  $P^{(n)}_{i}$ are such that $P^{(n)}_{i}\geq0$, the last term in \cref{proof:prop:max:eq} is negative and then the  proposition follows.
\end{proof}

Concluding, using  Proposition \ref{prop:max}, it is possible to give an upper bound for $S_{0}^{\myeps}(\tau_{E^n})$, in fact
\begin{align}
\bra{l}\tau_{E^n}\ket{l}&=\bra{l}\rho_{E}^{\otimes n}\ket{l}+\bra{l}\delta_{E^n}\ket{l}\nonumber\\
&\leq\bra{l}\rho_{E}^{\otimes n}\ket{l}+\frac{\etr}{2m_{0}}\left(1-\left(\frac{\lambda_{1}}{\lambda_{1}+\lambda_{0}}\right)^{n}\right) \label{eq:remember}.
\end{align}
Substituting in the formula above the actual values for the symmetric six-state protocol provided in equation \eqref{def:lambdas} the proof is concluded.

\subsection{Proof of $S_{2}(\eta_{X^n E^n})\geq \overline{S}_{2}^{\etr}(\rho_{XE}^{\otimes n}) -  \erk$}
\label{app:s2proof}
Using \cref{def:eta}, it follows that

\begin{align*}
	S_{2}(\eta_{X^{n}E^{n}})&=-\log\left[\tr_{X^{n}E^{n}}\left((\one\otimes\soproj )\tau_{X^{n}E^{n}} (\one\otimes \soproj)\right)^{2}\right]+ \\ &\quad\quad\quad+2\log\left(1-\frac{\erk}{2}\right).
\end{align*}

Using the first requirement in Definition~\ref{def:modsmoothrentwo}, the operator $\tau_{X^{n}E^{n}}$ is of the form
\begin{equation*}
\tau_{X^{n}E^{n}} =\frac{1}{2^n} \bigoplus_{x=0}^{2^{n}-1}\tau_{E^n}^{x}.
\end{equation*}

We concentrate on the argument of the logarithm in the first term on right-hand side of $S_{2}(\tau_{X^{n}E^{n}})$
\begin{align*}
\tr_{X^{n}E^{n}}&\left[\left((\one\otimes \soproj)\tau_{X^{n}E^{n}}(\one\otimes \soproj)\right)^{2}\right]\\&=\tr_{X^{n}E^{n}}\left[ \bigoplus_{x=0}^{2^{n}-1}\left(\frac{1}{2^n}\soproj\tau_{E^n}^{x}\soproj\right)^{2}\right]\\
&=\tr_{E^{n}}\left[ \sum_{x=0}^{2^{n}-1}\left(\frac{1}{2^n}\soproj\tau_{E^n}^{x}\soproj\right)^{2}\right]\\
&=\sum_{x=0}^{2^{n}-1}\tr_{E^{n}}\left[ \left(\frac{1}{2^n}\soproj\tau_{E^n}^{x}\soproj\right)^{2}\right]\\
&\leq\sum_{x=0}^{2^{n}-1}\tr_{E^{n}}\left[ \left(\frac{1}{2^n}\tau_{E^n}^{x}\right)^{2}\right]\\
&=\tr_{X^n E^{n}}\left[\left(\tau_{X^n E^n}\right)^{2}\right]
\end{align*}

Taking the first term of the Maclaurin expansion of $\log(1-\frac{\erk}{2})$ for $\erk$ small, we conclude that 
\begin{align*}
S_{2}(\eta_{X^{n}E^{n}})\geq S_{2}(\tau_{X^{n}E^{n}})-\erk.
\end{align*}
Using \cref{eq:def:modsmoothrentwo} the proof is concluded.

\end{appendix}

\bibliography{finitekey}

\end{document}